\documentclass[preprint,12pt]{amsart}
\usepackage[english]{babel}
\usepackage{amssymb}
\usepackage{amsmath}
\usepackage{geometry}
\usepackage{url}
\usepackage{color}

\newtheorem{lemma}{Lemma}

\newtheorem{proposition}{Proposition}
\newtheorem{theorem}{Theorem}

\theoremstyle{remark}
\newtheorem{remark}{Remark}

\def\ds{\displaystyle}

\def\bara{B\'ar\'any}

\def\cara{Carath\'eodory}

\def\R{\mathbb{R}}
\def\Q{\mathbb{Q}}
\def\Z{\mathbb{Z}}

\def\S{\mathbf{S}}
\def\K{\mathsf{K}}

\def\conv{\operatorname{conv}}
\def\supp{\operatorname{supp}}
\def\cone{\operatorname{cone}}

\def\com{\operatorname{comatrix}}
\def\zero{{\bf 0}}
\def\aff{\operatorname{aff}}
\def\pdcs{positively dependent colorful set}
\def\p{\boldsymbol{p}}
\def\x{\boldsymbol{x}}
\def\y{\boldsymbol{y}}
\def\z{\boldsymbol{z}}
\def\s{\boldsymbol{s}}
\def\t{\boldsymbol{t}}
\def\b{\boldsymbol{b}}
\def\c{\boldsymbol{c}}
\def\u{\boldsymbol{u}}
\def\r{\boldsymbol{r}}
\def\vv{\boldsymbol{v}}

\def\rk{\operatorname{rank}}
\def\ccp{{\sc Colorful \cara{}}}

\author{Fr\'ed\'eric Meunier}
\email{frederic.meunier@enpc.fr}

\author{Pauline Sarrabezolles}
\email{pauline.sarrabezolles@enpc.fr}

\title{Colorful linear programming, Nash equilibrium, and pivots}


\begin{document}

\maketitle 
 
 \begin{abstract}
The colorful \cara{} theorem, proved by \bara{} in 1982, states that given $d+1$ sets of points $\S_1,\ldots,\S_{d+1}$ in $\mathbb R^d$, with each $\S_i$ containing $\zero$ in its convex hull, there exists a set $T\subseteq\bigcup_{i=1}^{d+1}\S_i$ containing $\zero$ in its convex hull and such that $|T\cap\S_i|\leq 1$ for all $i\in\{1,\ldots,d+1\}$. An intriguing question -- still open -- is whether such a set $T$, whose existence is ensured, can be found in polynomial time. In 1997, \bara{} and Onn defined colorful linear programming as algorithmic questions related to the colorful \cara{} theorem. The question we just mentioned comes under colorful linear programming.

The traditional applications of colorful linear programming lie in discrete geometry. In this paper, we study its relations with other areas, such as game theory, operations research, and combinatorics. Regarding game theory, we prove that computing a Nash equilibrium in a bimatrix game is a colorful linear programming problem. We also formulate an optimization problem for colorful linear programming and show that as for usual linear programming, deciding and optimizing are computationally equivalent. We discuss then a colorful version of Dantzig's diet problem. We also propose a variant of the \bara{} algorithm, which is an algorithm computing a set $T$ whose existence is ensured by the colorful \cara{} theorem. Our algorithm makes a clear connection with the simplex algorithm and we discuss its computational efficiency. Related complexity and combinatorial results are also provided.
\end{abstract}

\section{Introduction}

\subsection{Context}
In 1982, \bara{} proved a colorful generalization of the \cara{} theorem, whose statement is the following.

\begin{theorem}[Colorful \cara{} theorem~\cite{Bar82}]
Given $d+1$ sets of points $\S_1,\ldots,\S_{d+1}$ in $\mathbb R^d$ such that each $\S_i$ contains $\zero$ in its convex hull, there exists a set $T$ of the form $\{s_1,\ldots,s_{d+1}\}$, with $s_i\in\S_i$ for every $i\in\{1,\ldots,d+1\}$, that contains $\zero$ in its convex hull.
\end{theorem}

A natural question raised by this theorem is whether such a {\em colorful} set $T$ can be computed in polynomial time. The case with $\S_1=\cdots=\S_{d+1}$, corresponding to the usual \cara{} theorem, is known to be solvable in polynomial time, via linear programming. However, the complexity of the colorful version remains an open question.

In 1997, \bara{} and Onn defined algorithmic and complexity problems related to the colorful \cara{} theorem~\cite{BO97}, giving birth to colorful linear programming. In their paper, the complexity question raised by the colorful \cara{} theorem is referred as an ``outstanding problem on the borderline of tractable and intractable problems''. In addition of providing a theoretical challenge, the colorful \cara{} theorem has several applications in discrete geometry (e.g. Tverberg partition, ``first selection lemma'', see~\cite{Mat02}). Any efficient algorithm computing such a colorful set $T$ would benefit these applications. 

A set of points is {\em positively dependent} if it is nonempty and contains $\zero$ in its convex hull. Given a configuration of $k$ sets of points $\S_1,\ldots,\S_k$ in $\R^d$, a set $T$ is {\em colorful} if it is of the form $\{s_1,\ldots,s_k\}$ with $s_i\in\S_i$ for every $i\in\{1,\ldots,k\}$. One of the problems studied in this paper is the following.\\

\noindent \ccp\\
{\bf Input.} A configuration of $d+1$ positively dependent sets of points $\S_1,\ldots,\S_{d+1}$ in $\Q^d$.\\
{\bf Task.} Find a positively dependent colorful set. \\

As we have already mentioned, the complexity status is still open. A more general problem has been recently proved to be PLS-complete by Mulzer and Stein~\cite{MuSt14}. The PLS class, where PLS stands for ``Polynomial Local Search'', contains the problems for which local optimality can be verified in polynomial time~\cite{JPY88}. The original proof of the colorful \cara{} theorem by \bara{} naturally provides an algorithm solving \ccp. This algorithm, known as the \bara{} algorithm, was analyzed and improved by \bara{} and Onn~\cite{BO97}. It is a pivot algorithm roughly relying on computing the closest facet of a simplex to $\zero$. Although not polynomial, this algorithm is quite efficient, as proved by Deza et al. through an extensive computational study~\cite{DHST08}. In addition to \ccp, \bara{} and Onn formulated the following problem, which is in a sense more general.\\

\noindent {\sc Colorful Linear Programming}\\
{\bf Input.} A configuration of $k$ sets of points $\S_1,\ldots,\S_k$ in $\Q^d$.\\
{\bf Task.} Decide whether there exists a positively dependent colorful set.\\

We emphasize that when we write ``{\sc Colorful Linear Programming}'', with small capital letters, we refer to that problem, as in \bara{}-Onn paper~\cite{BO97}, but when we write ``colorful linear programming'', we mean the study of the family of all problems about finding or deciding the existence of a \pdcs.

\bara{} and Onn showed that the case of {\sc Colorful Linear Programming} with $k=d$ is NP-complete even if each $\S_i$ is of size $2$, proving that the general case is NP-complete as well. It contrasts with \ccp. In this version, when each $\S_i$ is of size $2$, we clearly have a polynomial special case: select one point in each $\S_i$, find the linear dependency, and change for the other point in $\S_i$ for those having a negative coefficient.

A slightly more general version of \ccp{} can be defined with conic hulls instead of convex hulls.\\

\noindent\ccp{} (conic version)\\
{\bf Input.} A configuration of $d$ sets of points $\S_1,\ldots,\S_d$ in $\Q^d$ and an additional point $\boldsymbol{p}$ in $\bigcap_{i=1}^d\cone(S_i)\cap\Q^d$.\\
{\bf Task.} Find a colorful set $T$ such that $\boldsymbol{p}\in\cone(T)$.\\

The colorful set $T$ exists for sure because of a conic version of the colorful \cara{} theorem, also proved by \bara{} in the same paper~\cite{Bar82}. By an easy geometric argument, this problem coincides with \ccp{} when $\conv(\bigcup_{i=1}^d\S_i)$ does not contain $\zero$. Note that as usual for this kind of problems, there is a shift in the dimension when going from one version to the other.

 We define in the same way a conic version of {\sc Colorful Linear Programming}.\\

\noindent {\sc Colorful Linear Programming} (conic version)\\
{\bf Input.} A configuration of $k$ sets of points $\S_1,\ldots,\S_k$ in $\Q^{d}$ and an additional point $\boldsymbol{p}$ in $\Q^{d}$.\\
{\bf Task.} Decide whether there exists a colorful set $T$ such that $\boldsymbol{p}\in\cone(T)$.\\

\subsection{Main contributions}

\bigskip

\noindent{\bf ``{\sc Bimatrix} is a colorful linear programming problem''.} 

\smallskip

We prove that the problem {\sc Bimatrix}, consisting in computing a Nash equilibrium in a bimatrix game, is polynomially reducible to {\sc Finding Another Colorful Simplex Problem}, which is a colorful linear programming problem introduced by Meunier and Deza~\cite{MD12}. It shows that this latter problem is PPAD-complete. This was stated as an open question in the cited paper. On our way, we introduce a new method for proving that a linear complementarity problem belongs to the PPAD class, based on Sperner's lemma. This method seems to be interesting for its own sake, since it avoids the introduction of oriented primoids, as done usually. All these results are stated and proved in Section~\ref{sec:Nash}, where a definition of the PPAD class is also provided. \\


\noindent{\bf ``The simplex algorithm solves \ccp''.} 

\smallskip

In Section~\ref{sec:algo}, we show that the \bara{} algorithm and its improvement by \bara{} and Onn can be slightly modified in order to get what is more or less a ``Phase I'' simplex method. Instead of computing a closest facet at each pivot step, we select a new point by a classical reduced cost consideration. It simplifies the iteration, and improves the complexity of this latter. A complexity analysis of the overall algorithm is provided, as well as numerical experiments that show the effectiveness of the approach.\\  

\noindent{\bf ``Optimization and decision are equivalent for colorful linear programming''.} 

\smallskip

For the usual linear programming, it is known that the problem of deciding the existence of a solution to a linear program and the problem of optimizing a linear program are polynomially equivalent, see for instance Theorem 10.4 of Schrijver~\cite{Sch86}. In Section~\ref{sec:coloredOR}, we show that a similar equivalence holds for colorful linear programming. We present then an operations research type application of this result, namely a colored version of the famous diet problem by Dantzig.  We are not aware of similar applications in the literature, with industrial problems explicitly formulated as colorful linear programs.

\bigskip

In addition to these results, we show that {\sc Colorful Linear Programming} is NP-complete even if $k-d$ is fixed and each $\S_i$ is of size $2$. The proof is easy and this result should rather be called an observation. However, the complexity status of the case $k-d=1$ was a question of \bara{} and Onn (Section 5 of~\cite{BO97}) and it seems that people are not aware of the fact that the answer is easy. As a by-product of this complexity result, we get a new proof of the coNP-completeness of deciding whether a polytope is the projection of another polytope, both being described by systems of linear inequalities. These results are stated and proved in Section~\ref{sec:complexity}. 

We end the paper with a study of special cases and analogues of colorful linear programming in combinatorics (Section~\ref{sec:combin}). In particular, two combinatorial and polynomial cases of \ccp{} are presented. They are rather easy, but are not stated elsewhere and are clearly relevant for the understanding of the colorful linear programming landscape.

\section{Links with Nash equilibria}\label{sec:Nash}

\subsection{Another problem}\label{subsec:octahedral}

The fact that computing Nash equilibria is a colorful linear programming problem relies on the study of another problem similar to \ccp. This problem was proposed by Meunier and Deza~\cite{MD12} as a byproduct of an existence theorem, the Octahedron lemma~\cite{BM07,DHST06}, which by some features has a common flavor with the colorful \cara{} theorem. The Octahedron lemma states that if each $\S_i$ of the configuration is of size $2$ and if the points are in general position, the number of \pdcs s is even. By {\em general position}, we mean that no $d'+1$ points of the input points, $\zero$ included, lie in a same $(d'-1)$-dimensional affine subspace.

The problem we call {\sc Finding Another Colorful Simplex} is the following.\\

\noindent {\sc Finding Another Colorful Simplex}\\
{\bf Input.} A configuration of $d+1$ pairs of points $\S_1,\ldots,\S_{d+1}$ in $\Q^d$ and a positively dependent colorful set in this configuration.\\
{\bf Task.} Find another \pdcs.\\

 Another \pdcs{} exists for sure. Indeed, by a slight perturbation, we can assume that all points are in general position. If there were only one \pdcs{}, there would also be only one \pdcs{} in the perturbed configuration, which violates the evenness property stated by the Octahedron lemma. In their paper, Meunier and Deza question the complexity status of this problem. We solve the question by proving that it is actually a generalization of the problem of computing a Nash equilibrium in a bimatrix game. 

\subsection{{\sc Finding Another Colorful Simplex} is in \textnormal{PPAD}}

In~\cite{MD12}, it was noted that {\sc Finding Another Colorful Simplex} is in PPA. The class PPA, also defined by Papadimitriou in 1994~\cite{P94}, contains the class PPAD. PPA contains the problems that can be polynomially reduced to the problem of finding another degree $1$ vertex in a graph whose vertices all have degree at most $2$ and in which a degree $1$ vertex is already given. The graph is supposed to be implicitly described by the neighborhood function, which, given a vertex, returns its neighbors in polynomial time. The PPAD class is the subclass of PPA for which the implicit graph is oriented and such that each vertex has an outdegree at most $1$ and an indegree at most $1$. The problem becomes then: given an {\em unbalanced} vertex, that is a vertex $v$ such that $\deg^+(v)+\deg^-(v)=1$, find another unbalanced vertex.  See~\cite{P94} for further precisions.

We prove in this subsection that {\sc Finding Another Colorful Simplex} is in PPAD. We proceed by showing that the existence of the other \pdcs{} is a consequence of Sperner's lemma~\cite{Spe28}. 
Our method for proving that {\sc Finding Another Colorful Simplex} belongs to PPAD is adaptable for other complementarity problems, among them {\sc Bimatrix}. We believe that our method is new. It avoids the use, as in~\cite{P94,CDT09,kintali2009reducibility,VeVS12}, of {\em oriented primoids} or {\em oriented duoids} defined by Todd~\cite{Tod76}. 

One of the multiple versions of Sperner's lemma is the following theorem, proposed by Scarf~\cite{Sca67b}, which involves a triangulation of a sphere whose vertices are labeled. A simplex whose vertices get pairwise distinct labels is {\em fully-labeled}.

\begin{theorem}[Sperner's lemma]\label{thm:sperner}
Let $\mathsf{T}$ be a triangulation of an $n$-dimensional sphere $\mathcal{S}^n$ and let $V$ be its vertex set. Assume that the elements of $V$ are labeled according to a map $\lambda:V\rightarrow E$, where $E$ is some finite set. If $E$ is of cardinality $n+1$, then there is an even number of fully-labeled $n$-simplices.
\end{theorem}

We state now the main proposition of this subsection.

\begin{proposition}\label{prop:ppad}
{\sc Finding Another Colorful Simplex} is in \textup{PPAD}.
\end{proposition}

The proof uses the following fact. It is an elementary fact about feasible bases in linear programming, but since we were not able to find any place where it is clearly stated and proved, we do it for sake of completeness. Let us consider the system of linear equalities and inequalities of the form 
\begin{equation}\label{S}\tag{S}
A\x=\b,\;\x\geq\zero,
\end{equation}
where $A$ is an $m\times n$ real matrix of rank $m$ and let $\b\in\R^m$. This system is {\em nondegenerate} if any feasible solution has a support of size at least $m$.

\begin{lemma}[``Complements of feasible bases form a sphere'']\label{lem:bases}
Let $P$ be the set of feasible solutions of \eqref{S}. Suppose that \eqref{S} is nondegenerate and that $P$ is bounded. Then $\{\{1,\ldots,n\}\setminus\supp(\x):\;\x\in P\}$ is a simplicial complex that is isomorphic to a triangulation of an $(n-m-1)$-dimensional sphere $\mathcal{S}^{n-m-1}$.
\end{lemma}

\begin{proof}
Denote by $\K$ the simplicial complex $\{\{1,\ldots,n\}\setminus\supp(\x):\;\x\in P\}$. It is a simplicial complex because the system is nondegenerate. The vertices of $\K$ are in one-to-one correspondence with the facets of $P$, which are in one-to-one correspondence with the vertices of the polar $P^*$ of $P$. The resulting one-to-one correspondence between the vertices of $\K$ and the vertices of $P^*$ leads to a simplicial map between $\K$ and the face complex of $\partial P^*$ (here we use the correspondence between the faces of $P$ and $P^*$ as stated in Proposition 5.3.5 of~\cite{Mat02}). The inverse of this map is also simplicial. Thus $\K$ is isomorphic to the face complex of $\partial P^*$.
\end{proof}

\begin{proof}[Proof of Proposition~\ref{prop:ppad}]
Define $A$ as the $(d+1)\times(2d+2)$ matrix, whose columns are the coordinates of the input points, with an additional $1$ at the $(d+1)$th position. This is done ``with repetitions'' and the columns are numbered so that the points in $\S_i$ correspond to columns $2i-1$ and $2i$ (thus, if a point is in several $\S_i$'s, it will provide several columns). Define $\b$ as the vector with $0$ at every position except at the $(d+1)$th position, where we also have a $1$. By a perturbation argument -- see~\cite{MeCh89} for instance for a description of such a polynomial-time computable perturbation -- we can assume that the system $$A\x=\b,\;\x\geq\zero$$ is nondegenerate. Every feasible basis of this system that contains exactly one element of each $\{2i-1,2i\}$ -- which we call a {\em colorful feasible basis} -- provides a \pdcs. We describe now a polynomial reduction of the problem of finding a colorful basis in this system to the problem of finding a fully-labeled simplex in a triangulation of a $d$-dimensional sphere.


We define a simplicial complex $\K$ with vertex set $\{1,\ldots,2d+2\}$ and whose simplices $\sigma$ are such that $\{1,\ldots,2d+2\}\setminus\sigma$ is the support of a feasible solution of the system. $\K$ is a simplicial complex of the form stated in Lemma~\ref{lem:bases}. It is thus a triangulation of $\mathcal{S}^{d}$, with $m=d+1$ and $n=2d+2$. Now, for $v$ a vertex of $\K$, define $\lambda(v)$ to be its color, i.e., the index $i$ such that $v\in\{2i-1,2i\}$. Any fully-labeled simplex $\sigma$ of $\K$ is such that $\{1,\ldots,2d+2\}\setminus\sigma$ is a colorful feasible basis and conversely. There is thus an explicit one-to-one correspondence between the fully-labeled simplices of $\K$ and the colorful feasible bases. 
The triangulation $\K$ here can easily be encoded by a Turing machine computing the neighbors of any simplex in the triangulation in polynomial time. There is a proof of Sperner's lemma (Theorem~\ref{thm:sperner}) via an oriented path-following argument~\cite{Sca67b,Meu13}, which considers directed paths joining fully-labeled simplices. Given a fully-labeled simplex, finding another fully-labeled simplex is thus in \textup{PPAD}, and so is {\sc Finding Another Colorful Simplex}.
\end{proof}


\subsection{Reduction of {\sc Bimatrix}}\label{subsec:bimatrix}  

A bimatrix game involves two $m\times n$ matrices with real coefficients $A=(a_{ij})$ and $B=(b_{ij})$. There are two players. The first player chooses a probability distribution on $\{1,\ldots,m\}$, the second a probability distribution on $\{1,\ldots,n\}$. Once these probability distributions have been chosen, a pair $(\bar{i},\bar{j})$ is drawn at random according to these distributions. The first player gets a payoff equal to $a_{(\bar{i},\bar{j})}$ and the second a payoff equal to $b_{(\bar{i},\bar{j})}$. A Nash equilibrium is a choice of distributions in such a way that if a player changes his distribution, he will not get in average a strictly better payoff.


Let $\Delta^k$ be the set of vectors $\boldsymbol{x}\in\mathbb{R}_+^k$ such that $\sum_{i=1}^kx_i=1$. Formally, a {\em Nash equilibrium} is a pair $(\boldsymbol{y}^*,\boldsymbol{z}^*)$ with $\boldsymbol{y}^*\in\Delta^m$ and $\boldsymbol{z}^*\in\Delta^n$ such that
\begin{equation}\label{eq:mixedNash}\boldsymbol{y}'^TA\boldsymbol{z}^*\leq \boldsymbol{y}^{*T}A\boldsymbol{z}^*\mbox{ for all $\boldsymbol{y}'\in\Delta^m$}
\quad\mbox{ and }\quad\boldsymbol{y}^{*T}B\boldsymbol{z}'\leq \boldsymbol{y}^{*T}B\boldsymbol{z}^*\mbox{ for all $\boldsymbol{z}'\in\Delta^n$}.
\end{equation} 
It is well-known that if the matrices have rational coefficients, there is a Nash equilibrium with rational coefficients, which are not too large with respect to the input. {\sc Bimatrix} is the following problem: given $A$ and $B$ with rational coefficients, find a Nash equilibrium. 
Papadimitriou showed in 1994 that {\sc Bimatrix} is in {\textup PPAD}~\cite{P94}. Later, Chen, Deng, and Teng~\cite{CDT09} proved its PPAD-completeness.

A combinatorial approach to these equilibria consists in studying the {\em complementary} solutions of the two systems
\begin{equation}\label{eq:A}
[A,I_m]\boldsymbol{x}=(1,\ldots,1)^T\mbox{ and }\boldsymbol{x}\in\mathbb{R}_+^{n+m}
\end{equation} and
\begin{equation}\label{eq:B}
[I_n,B^ T]\boldsymbol{x}=(1,\ldots,1)^T\mbox{ and }\boldsymbol{x}\in\mathbb{R}_+^{n+m}.
\end{equation} By {\em complementary solutions}, we mean a solution $\boldsymbol{x}_A$ of~\eqref{eq:A} and a solution $\boldsymbol{x}_B$ of~\eqref{eq:B} such that $\boldsymbol{x}_A\cdot\boldsymbol{x}_B=0$. Indeed, complementary solutions with $\supp(\boldsymbol{x}_A)\neq\{n+1,\ldots,n+m\}$ or $\supp(\boldsymbol{x}_B)\neq\{1,\ldots,n\}$ give a Nash equilibrium. This point of view goes back to Lemke and Howson~\cite{LeHo64}. A complete proof within this framework can be found in Remark 6.1 of~\cite{Meu13}.

We derive the difficulty of {\sc Finding Another Colorful Simplex} from the complexity of {\sc Bimatrix}.

\begin{proposition}\label{prop:another}
{\sc Finding Another Colorful Simplex} is \textup{PPAD}-complete.
\end{proposition}

\begin{proof}

We prove that the following version of {\sc Finding Another Colorful Simplex} with cones is PPAD-complete. This version is equivalent to {\sc Finding Another Colorful Simplex}.\\

\noindent {\sc Finding Another Colorful Cone}\\
{\bf Input.} A configuration of $d+1$ pairs of points $\S_1,\ldots,\S_{d+1}$ in $\Q^{d+1}$, an additional point $\boldsymbol{p}$ in $\Q^{d+1}$ such that $\conv(\bigcup_{i=1}^{d+1}\S_i)$ does not contain $\zero$, and a colorful set $T$ such that $\boldsymbol{p}\in\cone(T)$.\\
{\bf Task.} Find another colorful set $T'$ such that  $\boldsymbol{p}\in\cone(T')$.\\

The proof uses a reduction of {\sc Bimatrix} to {\sc Finding Another Colorful Cone}. Consider an instance of {\sc Bimatrix}. First note that we can assume that all coefficients of $A$ and $B$ are positive. Indeed, adding a same constant to all entries of the matrices does not change the game. Build the $(m+n)\times(2(m+n))$ matrix
$$M=\left(\begin{array}{cccc}A & I_m & \zero & \zero \\ \zero & \zero & I_n & B^T\end{array}\right).$$ We denote by $M_i$ the $i$th column of $M$. Note that the vector $\boldsymbol{u}=(1,\ldots,1)\in\mathbb{R}^{n+m}$ is in the conic hull of $T=\{M_{n+1},\ldots,M_{n+m},M_{n+m+1},\ldots,M_{2n+m}\}$.
Indeed, the corresponding submatrix is the identity matrix.

Let $\S_i$ be the pair $\{M_i,M_{m+n+i}\}$ for $i=1,\ldots,m+n$. Since all coefficients of $A$ and $B$ are positive, $\zero$ is not in the convex hull of the columns of $M$ and $\boldsymbol{u}$. A polynomial time algorithm solving {\sc Finding Another Colorful Simplex} with $T$ as input set would find another colorful set $T'$ such that $\boldsymbol{u}\in\cone(T')$. The decomposition of $\boldsymbol{u}$ on the points in $T'$ gives a nonnegative vector $\boldsymbol{x}$ such that $M\boldsymbol{x}=\boldsymbol{u}$, $x_ix_{m+n+i}=0$ for $i=1,\ldots,m+n$, and $\supp(\boldsymbol{x})\neq\{n+1,\ldots,2n+m\}$. Such an $\boldsymbol{x}$ can be written $(\boldsymbol{x}_A,\boldsymbol{x}_B)$ with $\boldsymbol{x}_A,\boldsymbol{x}_B\in\mathbb{R}_+^{m+n}$ satisfying $\boldsymbol{x}_A\cdot\boldsymbol{x}_B=0$ and such that $\supp(\boldsymbol{x}_A)\neq\{n+1,\ldots,n+m\}$ or $\supp(\boldsymbol{x}_B)\neq\{1,\ldots,n\}$. In other words, it would find a Nash equilibrium. {\sc Bimatrix} being PPAD-complete, Proposition~\ref{prop:ppad} implies therefore that {\sc Finding Another Colorful Simplex} is PPAD-complete.
\end{proof}

This proof shows that {\sc Finding Another Colorful Simplex} is more general than computing complementary solutions of Equations~\eqref{eq:A} and~\eqref{eq:B}. In~\cite{MD12}, a pivoting algorithm for solving {\sc Finding Another Colorful Simplex} is proposed. It reduces to the classical pivoting algorithm due to Lemke and Howson~\cite{LeHo64} used for computing such complementary solutions.


\begin{remark}
A reviewer noted that Propositions~\ref{prop:ppad} and~\ref{prop:another} are close to results presented in a paper by Kir\'aly and Pap~\cite{KiPa13}. Actually, using results of the latter paper, we can directly prove Proposition~\ref{prop:ppad}. As suggested by the referee, consider the set of all $(2d+2)$-dimensional vectors corresponding to coefficients of positively dependences. It is a $(d+1)$-dimensional polytope $Q$ with exactly $(2d+2)$ facets, each facet corresponding to a point from the input with a $0$ coefficient in the positive dependence. Color each facet with the color of the corresponding input point. The \pdcs{}s correspond to vertices of $Q$ surrounded by facets of distinct colors (``panchromatic vertices''). The Theorem 2 of their paper ensures that if there is a panchromatic vertex, there is another one, thus showing the Octahedron lemma, and the related computational problem is implicitly proved to be PPAD. We get thus Proposition~\ref{prop:ppad} by this way. We did not find a direct and easy way for proving Proposition~\ref{prop:another} with the results of the mentioned paper, but it should be possible.

However, we think that the proof we propose for Proposition~\ref{prop:ppad} remains interesting, especially because it describes a systematic way for proving PPAD-membership of linear complementary problems. Moreover, it reveals that an elementary Sperner-type computational problem is PPAD-complete, see next remark.
\end{remark}

\begin{remark}[Complexity of Sperner's lemma]\label{rk:sperner}

Consider the problem\\

\noindent{\bf Input.} A triangulation $\mathsf{T}$ of $\mathcal{S}^d$ involving $2(d+1)$ vertices, a labeling $\lambda:V(\mathsf{T})\rightarrow\{1,\ldots,d+1\}$, and a fully-labeled simplex.\\
\noindent{\bf Task.} Find another fully-labeled simplex. \\

The reduction used in the proof of Proposition~\ref{prop:ppad} and Proposition~\ref{prop:another} shows that this problem is actually PPAD-complete, even if each label appears exactly twice. Sperner-type problems have already been proved to be PPAD-complete~\cite{CD06,P94}, but these latter problems are in fixed dimension, with an exponential number of vertices, and with a labeling given by an oracle, while the Sperner-type problem we introduce has an explicit description of the vertices and of the labeling. Note that the number of vertices is small. As already mentioned, the paper by Kir\'aly and Pap~\cite{KiPa13} present and discuss computational problems that are in the same spirit and also proved to be PPAD-complete, but they are presented in a form that makes them farther from the classical Sperner's lemma: their formulation involves an unbounded polyhedron, with a coloring of the facets, and some conditions on the extreme rays. 

Remark~\ref{rk:sperner_comp} in Section~\ref{sec:combin} will exhibit some polynomial cases of the {\textsc Sperner}-type problem we introduce here.
\end{remark}

\section{Simplexification of \bara{}'s algorithm}\label{sec:algo}

The proof of the colorful \cara{} theorem by \bara{} \cite{Bar82} actually provides an algorithm. This algorithm was improved and analyzed by \bara{} and Onn in 1997. In this section, we propose a new algorithm to compute a colorful solution. It can be seen as a ``simplexification'' of \bara{}'s algorithm. This latter requires at each iteration to compute the projection of the origin on a simplex. We show that this projection operation, which is not an easy task, although polynomial, can be replaced by a simple reduced cost test. This leads to an algorithm that is almost the classical simplex algorithm.

\subsection{\bara's algorithm}

The pivoting algorithm proposed by \bara{} for solving \ccp{} goes as follows. The input is the sets $\S_1,\ldots,\S_{d+1}$ of points in $\Q^d$, each of cardinality $d+1$ and positively dependent. \\

\noindent{\bf \bara's algorithm}
\begin{itemize}
\item Choose a first colorful set $T_1$ of size $d+1$, let $i:=0$, and let $\x_1$ be the point of minimum norm in $\conv(T_1)$.
\item Repeat: 
\begin{itemize}
\item Let $i:=i+1$.
\item If $\zero\in\conv(T_i)$, stop and output $T_i$.
\item Otherwise, find $\x_i$ of minimum norm in $\conv(T_i)$; choose $\t_i\in T_i$ such that $\x_i\in\conv(T_i\setminus\{\t_i\})$; choose a point $\t$ in $\S_k$ minimizing $\t\cdot\x_i$, where $k$ is the color of $\t_i$; define $T_{i+1}:=T_i\setminus\{\t_i\}\cup\{\t\}$.
\end{itemize}
\end{itemize} \bigskip


It is rather straightforward to prove that the sequence $(\|\x_i\|)$ converges to $0$ when $i$ goes to infinity, see~\cite{BO97} for the details. By finiteness, it implies that \bara{}'s algorithm finds a \pdcs{} in finite time. The technical step is the computation of $\x_i$, which is a projection computation. It is a rather heavy task, which can moreover only be solved approximately. The improvement of \bara{} and Onn consists in replacing this projection by the computation of another point that plays a similar role, but which requires only linear algebra for its computation. Deza et al.~\cite{DHST08} shows how to compute this point in $O(d^3)$. They also proceed to an extensive experimental study of algorithms solving this problem. In addition to some heuristics, ``multi-update'' versions are also proposed, but they do not avoid this kind of operations.

Our algorithm, described in the next section, allows to do each iteration in $O(d^2)$. Perhaps more interestingly, it shows that in some sense the simplex algorithm solves \ccp{}, with almost all classical pivot-rules as possible variations.

\subsection{The simplex version}

\subsubsection{The algorithm}

We add a dummy point $\boldsymbol{v}$ and define the following optimization problem. 
\begin{equation}\tag{P}\label{prog}\begin{array}{rl}\min & z \\
\mbox{s.t.} & \left(\begin{array}{ccc}\multicolumn{3}{c}M \\ 1 & \ldots & 1\end{array}\right)\boldsymbol{\lambda}+z\bar{\boldsymbol{v}}=\left(\begin{array}{c}0 \\ \vdots \\ 0 \\ 1\end{array}\right)\\
& \boldsymbol{\lambda}\geq\zero,\,z\geq 0,
\end{array}
\end{equation}
where $\bar{\boldsymbol{v}}=(\boldsymbol{v},1)$ and where $M$ is the $d\times(d+1)^2$ matrix whose columns are the points of $\bigcup_{i=1}^{d+1}\S_i$ (where as in the proof of Proposition~\ref{prop:ppad} we consider the disjoint union of the points). This optimization problem simply looks for an expression of $\zero$ as 
a convex combination of the points in $\{\boldsymbol{v}\}\cup\bigcup_{i=1}^{d+1}\S_i$ with a minimal weight on $\boldsymbol{v}$. Especially, if $\zero\in\conv(\bigcup_i\S_i)$, the optimal value is $0$. The idea consists in seeking an optimal basis, with the terminology of linear programming, which in addition is required to be colorful. The colorful \cara{} theorem ensures that such a basis exists.

Now, choose a first {\em transversal} $F_1$, which is a colorful set of cardinality $d$. Choose the dummy point $\boldsymbol{v}$ so that $F_1\cup\{\boldsymbol{v}\}$ contains $\zero$ in the interior of its convex hull (we explain below in Section~\ref{subsubsec:pre} how it is possible to assume the existence of such set and point, and how to compute them). Note that $F_1\cup\{\boldsymbol{v}\}$ is a feasible basis. The algorithm proceeds with simplex pivots, going from feasible colorful basis to feasible colorful basis, until an optimal colorful basis is found. We start with $i:=0$. We repeat then 

\begin{itemize}\item[\null]\begin{itemize}
\item Let $i:=i+1$.
\item Choose a point $\boldsymbol{t}$ of the missing color in $F_i$ with negative reduced cost. The reduced costs are computed according to the current basis $F_i\cup\{\boldsymbol{v}\}$.
\item Proceed to a simplex pivot operation with $\boldsymbol{t}$ entering the current basis. 
\item If $\boldsymbol{v}$ leaves the basis, stop and output $F_i\cup\{\boldsymbol{t}\}$ (it is an optimal colorful basis).
\item Otherwise, define $F_{i+1}$ to be the new basis minus $\boldsymbol{v}$.
\end{itemize} \end{itemize}\bigskip

The following lemma shows that as long as a \pdcs{} has not been found, selecting a point of a color missed by $F_i$ with a negative reduced cost is the same as selecting a point in the open half-space delimited by $\aff(F_i)$ and containing $\zero$. As in \bara's algorithm, we get closer and closer to the origin, and by finiteness, we eventually find a \pdcs.

\begin{lemma}\label{lem:pivot} The points in the open half-space delimited by $\aff(F_i)$ and containing $\zero$ are precisely the points with a negative reduced cost.
\end{lemma}

\begin{proof}
Let $F_i=\{\boldsymbol{u}_1,\ldots,\boldsymbol{u}_d\}$ and let $\boldsymbol{t}$ be any other point in $\left(\bigcup_{j=1}^{d+1}\S_j\right)\setminus F_i$. Consider $x_1,\ldots,x_d,r,s\in\mathbb{R}$ such that 
\begin{equation}\label{eq:pivot}\boldsymbol{t}+s\boldsymbol{v}+\sum_{i=1}^dx_i\boldsymbol{u}_i=\zero\qquad\mbox{and}\qquad 1+s+\sum_{i=1}^dx_i=0.\end{equation} The reduced cost of $\boldsymbol{t}$ is exactly $s$. Therefore, proving the lemma amounts to prove that $s$ is negative exactly when $\boldsymbol{t}$ is in the open half-space delimited by $\aff(F_i)$ and containing $\zero$. To see this, note that \eqref{eq:pivot} implies \begin{equation}\label{eq:pivot_bis}(\boldsymbol{t}-\boldsymbol{u}_1)+s(\boldsymbol{v}-\boldsymbol{u}_1)+\sum_{i=2}^dx_i(\boldsymbol{u}_i-\boldsymbol{u}_1)=\zero.\end{equation} Now, choose a unit vector $\boldsymbol{n}$ orthogonal to $\aff(F_i)$. Take the scalar product of Equation~\eqref{eq:pivot_bis} and $\boldsymbol{n}$. It gives $$r\boldsymbol{n}\cdot(\boldsymbol{t}-\boldsymbol{u}_1)+s\boldsymbol{n}\cdot(\boldsymbol{v}-\boldsymbol{u}_1)=0$$ and the conclusion follows since $\boldsymbol{v}$ and $\zero$ are in the same half-space delimited by $\aff(F_i)$.
\end{proof}

This approach is reminiscent of the ``Phase I'' simplex method, which computes a first feasible basis by solving an auxiliary linear program whose optimal value is $0$ on such a basis.

\subsubsection{Preprocessing}\label{subsubsec:pre}

\noindent {\bf Computation of $F_1$.}
The algorithm requires to have a first transversal $F_1$ whose elements are linearly independent. If there is no such transversal, it means by Rado's theorem (\cite{Ra42}, see also p.702 of~\cite{Sch03}) that there exists $I\subseteq[d]$ such that $\rk(\bigcup_{i\in I}\S_i)<|I|$. Deciding the existence of the transversal and computing $I$ if it does not exist can be done in polynomial time with the help of Edmonds' matroid intersection algorithm. Now, we remove all colors not in $I$ and consider \ccp{} in the linear vector space spanned by $\bigcup_{i\in I}\S_i$. We can thus repeat this process several times if necessary  and we will eventually get an input with a transversal whose elements are linearly independent. \\

\noindent{\bf Computation of $\vv$.}  It is easy to choose $\vv$. A way to do it that makes the complexity computation required for Proposition~\ref{prop:complex_simplex} below easy is to define $\vv$ as $-\sum_{i=1}^d\u_i$, where the $\u_i$'s are the vertices of $F_1$.

\subsubsection{Complexity analysis}

As for the classical simplex algorithm, it is not easy to provide a complexity analysis. We can however prove the following result.

\begin{proposition}\label{prop:complex_simplex}
Suppose that the $\S_i$'s and $\zero$ are in general position. Suppose moreover that all points have integer coordinates. Then the algorithm computes a colorful set $T$ in at most $(d^8+1)\Delta^4$ iterations, where $\Delta$ is the largest absolute value of any $d\times d$ subdeterminant of the input points (i.e., of $M$).
\end{proposition}

\begin{proof}
Let us write the program~\eqref{prog} under the form 
$$\begin{array}{lll}\min & \c\cdot\x \\
\mbox{s.c.} & A\x=\b \\
& \x\geq\zero,
\end{array}$$
where $$\b=\c=\left(\begin{array}{c}0 \\ \vdots \\ 0 \\1\end{array}\right)\qquad\mbox{ and }\qquad A=\left(\begin{array}{cccc}\multicolumn{3}{c}M & -\sum_{i=1}^d\u_i\\ 1 & \ldots & 1 & 1\end{array}\right).$$
Let $\delta_P$ and $\gamma_P$ be respectively the minimum and maximum values of all positive elements of basic feasible solutions of the program~\eqref{prog}. For a feasible basis $B$, denote by $\s^B$ the corresponding reduced costs. Define then $$\delta'_D=\min\{-s_j^B:\;B\mbox{ is a feasible basis and }s_j^B<0\}$$ and 
$$\gamma'_D=\max\{-s_j^B:\;B\mbox{ is a feasible basis and }s_j^B<0\}.$$ A theorem by Kitahara and Minzuno~\cite{KiMi13} states that the number of different basic feasible solutions generated by the algorithm, and thus the number of pivot steps when the input is nondegenerate, is at most $$\left\lceil d\frac{\gamma_P\gamma'_D}{\delta_P\delta'_D}\right\rceil.$$ In order to apply this theorem, we bound $\delta_P$, $\gamma_P$, $\delta_D$, and $\gamma_D$.

For a feasible basis $B$, the corresponding feasible basic solution is given by $$\x_B=A_B^{-1}\b\qquad\mbox{and}\qquad\x_N=\zero.$$ We have $A_B^{-1}=\frac 1 {\det A_B}(\com A_B)^T$. Note that the last column of $A$ is always in any feasible basis $B$ generated by the algorithm, otherwise we would be at the optimal solution of \eqref{prog}. We compute $\det(A_B)$ by expanding along the last row. Since the last column of $A$ is the opposite of the sum of $d$ other columns, we get that $\det(A_B)\leq(d^2+1)\Delta$. Using that the input points have integer coordinates, we get $$\delta_P\geq\frac 1 {(d^2+1)\Delta}.$$ 
The input points having integer coordinates we have $\det A_B\geq 1$. Using Cramer's formula, we get $$\gamma_P\leq d\Delta.$$
We have $$\s_B^B=\zero\qquad\mbox{and}\qquad\s_N^B=\c_N-A_N^T(A_B^T)^{-1}\c_B.$$ Here, $\c_N=\zero$ and $\c_B$ is the $d$th standard unit vector, except for the final basis computed by the algorithm. Using exactly the same approach as for $\delta_P$ and $\gamma_p$ to bound the entries of $A_N^T(A_B^T)^{-1}$, we get that 
$$\delta'_D\geq \frac 1 {(d^2+1)\Delta}\qquad\mbox{ and }\qquad\gamma'_D\leq(d^2+1)\Delta.$$
\end{proof}

We are not able to propose a convincing way to compare the theoretical complexity of \bara-Onn's algorithm with ours. However, our algorithm works also for the conic version of \ccp, and in this case, our analysis provides the following result. The first transversal can be assumed to exist using a similar approach as in Section~\ref{subsubsec:pre}: by applying Rado's theorem, we can assume that there exists a linearly independent colorful set; the first transversal is given by any $(d-1)$-subset $F_1$ of this independent colorful set whose linear span does not contain $\p$. 

\begin{proposition}
Consider \ccp{} in its conic version. Suppose that the matrix $M$ whose columns are the $d^2$ input points is totally unimodular and that $\p$ is in $\Z^d$. If the system $$M\x=\p,\;\x\geq\zero$$ is nondegenerate, then the algorithm computes a colorful cone containing $\p$ in at most $d^6(\|\p\|_{\infty}+1)^4$ pivot steps.
\end{proposition}

\begin{proof}
We define $\boldsymbol{v}$ to be $\p-\sum_{i=1}^{d-1}\u_i$, where the $\u_i$'s form the subset $F_1$. We have then 
$$\delta_P\geq  \frac 1 {d(\|\p\|_{\infty}+1)}\qquad\mbox{and}\qquad\gamma_P\leq d^2\|\p\|_{\infty},$$
and 
$$\delta'_D\geq \frac 1 {d(\|\p\|_{\infty}+1)}\qquad\mbox{ and }\qquad\gamma'_D\leq d(\|\p\|_{\infty}+1).$$

\end{proof}

\subsubsection{Numerical results}

We implemented our algorithm in C++. The tests are performed on a PC Intel{\small\textsuperscript{\textregistered}} Core{\small\textsuperscript{\texttrademark} i3-2310M}, with two 64-bit CPUs, clocked at 2.1 GHz, with 4 GB RAM.
The instances are provided by five random generators, implemented by Huang in MATLAB and described in his master thesis~\cite{Hu07}. All the generators provide instances of $(d+1)^2$ points in general position on the unit sphere, partitioned into $d+1$ colors and such that the origin $\zero$ is in the convex hull of each color. Descriptions of the generators can be found in~\cite{Hua05}. At each iteration, we choose the entering point $\boldsymbol{t}$ that has the most negative reduced cost.

Table~\ref{tab:results} presents the computational results on 50 instances by dimension and by generators (the ``tube'' instances are those referred as ``unbalanced'' in Huang's master thesis). The columns ``time'' give the average execution time of the algorithm in milliseconds. The columns ``\# pivots'' give the average number of pivots. The entry corresponding to the ``tube'' case in dimension $384$ is empty, since we faced cycling behavior for some instances (we felt that adding anti-cycling pivot rules was not imperative for our experiments).

\begin{table}[!h]
\begin{center}
\begin{tiny}\begin{tabular}{r|r|r|r|r|r|r|r|r|r|r}
 & \multicolumn{2}{c|}{Random}  & \multicolumn{2}{c|}{Tube} & \multicolumn{2}{c|}{Highdensity}  & \multicolumn{2}{c|}{Lowdensity}  & \multicolumn{2}{c}{Middensity} \\
Dimension&time&\# pivots&time&\# pivots&time&\# pivots&time&\# pivots&time&\# pivots\\ \hline

   3 &0.0135 &1.94& 0.0123 &2.02& 0.0342&1.62&0.0138&2.32 & 0.0170&1.70\\

   6 &0.0180&3.38& 0.0195& 3.42& 0.0474&1.98&0.0213&6.50& 0.0164&2.88\\

   12 &0.0406 &6.56&0.0396 &7.68& 0.0609&1.84&0.0591&19.00& 0.0371&4.88\\

   24 & 0.1433&13.76& 0.1574&19.66& 0.0871&1.94&0.2958&51.06&0.1123&9.62\\

   48 & 0.9612 &31.86& 1.1684&43.88& 1.1006&1.94&2.7946&133.70& 0.7725&19.44\\

   96 & 8.5069&76.42& 11.2441&108.10 &3.1116 &1.92&28.5813 &349.44&6.1306&39.46\\

   192 & 81.0170 &186.62& 250.1050&284.96&21.6753 &1.86&263.1400&831.26&50.2998&93.88\\ 

   384 &1111.5020&476.50&	&	&441.1310&2.00	&5987.8880&2032.60&846.9148&279.12 \\  \multicolumn{1}{c}{}\\ \multicolumn{1}{c}{}
\end{tabular}

\end{tiny}
\caption{
 \label{tab:results}Average solution time and number of pivots for the simplex-like algorithm}
\end{center}
\end{table}

We compare these results with those of the \bara{}-Onn algorithm presented in the paper by Deza et al.~\cite{DHST08} using the same generators. We provide the detailed numerical results of that paper in Table~\ref{tab:huang}, taken from Huang's master thesis (in which the algorithm is referred to as BO2).

\begin{table}[!h]
\begin{center}

\begin{tiny}\begin{tabular}{r|r|r|r|r|r|r|r|r|r|r}
 & \multicolumn{2}{c|}{Random}  & \multicolumn{2}{c|}{Tube} & \multicolumn{2}{c|}{Highdensity}  & \multicolumn{2}{c|}{Lowdensity}  & \multicolumn{2}{c}{Middensity} \\
Dimension & time &\# pivots&time&\# pivots&time&\# pivots&time&\# pivots&time&\# pivots\\ \hline

   3   & 0.7501& 2.18  & 1.083\hspace{1.15mm} & 3.70  & 0.5735 &1.53 & 0.9090 & 2.88 & 0.7613 & 2.24\\

   6   & 1.832\hspace{1.15mm} & 4.94  & 3.716\hspace{1.15mm} & 11.46   & 0.8450 &1.73 & 2.404\hspace{1.15mm} & 6.81 & 1.936\hspace{1.15mm} & 5.26\\

   12  & 5.789\hspace{1.15mm} & 11.15 & 13.72\hspace{2.3mm} & 29.10   & 1.457\hspace{1.15mm}  & 1.61 & 7.225\hspace{1.15mm} & 14.26 & 5.984\hspace{1.15mm}  & 11.51\\

   24  & 21.88\hspace{2.3mm} & 24.94 & 56.23\hspace{2.3mm}  & 69.98  & 3.685\hspace{1.15mm} & 1.65 & 24.91\hspace{2.3mm} & 28.53 & 21.38\hspace{2.3mm} & 23.76\\

   48  & 96.67\hspace{2.3mm} & 52.32 & 251.7\hspace{3.45mm}   & 150.2\hspace{1.15mm}  & 15.89\hspace{2.3mm} & 1.48 & 105.8\hspace{3.45mm} & 59.9\hspace{1.15mm}  & 92.40\hspace{2.3mm} & 47.9\hspace{1.15mm}  \\

   96  & 513.1\hspace{3.45mm} & 107.9\  & 1278\hspace{5.6mm} & 298.9\hspace{1.15mm}   & 82.79\hspace{2.3mm} & 1.58 & 542.5\hspace{3.45mm}  & 111.0\hspace{1.15mm}  & 488.1\hspace{3.45mm} & 95.92\\

   192 & 3465\hspace{5.6mm} & 214.7\  & 9253\hspace{5.6mm} & 625.8\hspace{1.15mm} & 559.7\hspace{3.45mm} &1.8\hspace{1.15mm} & 3711\hspace{5.6mm}  &221.6\hspace{1.15mm} & 3274\hspace{5.6mm}  & 191.9\hspace{1.15mm}  \\ 

   384 & 34230\hspace{5.6mm} & 433.6\hspace{1.15mm} &	             &	             & 4355\hspace{5.6mm}  & 1.64 & 37010\hspace{5.6mm}  &459.8\hspace{1.15mm} & 30150\hspace{5.6mm} & 384\hspace{3.3mm}  \\  \multicolumn{1}{c}{}\\ \multicolumn{1}{c}{}
\end{tabular}

\end{tiny}
\caption{
 \label{tab:huang}Average solution time and number of pivots for the \bara{}-Onn algorithm, as recorded in~\cite{Hu07}}
\end{center}
\end{table}

In general, our number of pivots is slightly larger than what they get. Regarding the computation time, it is hard to draw a conclusion since their implementation was done in MATLAB and since they used a different machine (a server with eight 64-bit CPUs, clocked at 2.6 GHz, with 64 GB RAM). However, we are always much faster -- our total computation time is of the order of ten times smaller than theirs -- and our time per iteration is up to thirty times smaller.


\section{Optimization in colorful linear programming and applications}\label{sec:coloredOR}

\subsection{Optimization}

The colorful linear programming problem is defined as a feasibility problem. However, it is natural to introduce an optimization version. For the usual linear programming problem, feasibility and optimization are equivalent: if we have a method for deciding whether a linear program is feasible, then we can use it for solving linear programs to optimality. The same holds for colorful linear programming, as we are now going to prove.

A {\em colorful linear program} is a mathematical program of the form 
\begin{equation}\tag{Q}\label{P}
\begin{array}{lll}
\min &\c\cdot\x\\
\mbox{s.t.}& A\x=\b\\
&\x\geq \zero\\
&|\supp(\x)\cap I_i|\leq 1,\qquad\forall i\in\{1,\ldots,k\}
\end{array}
\end{equation}
where $A$ is a $d\times n$ real matrix, $\b$ an element of $\R^d$, and $\c$ an element of $\R^n$,
 and where the $I_i$'s form a partition of the columns of $A$, the number of colors being $k$. The problem of deciding whether a colorful linear program has a solution is exactly {\sc Colorful Linear Programming} in its conic version.

\begin{proposition}\label{prop:feas_opt}
 The two problems of deciding the feasibility of a colorful linear program and of finding the optimal solution of a colorful linear program are polynomially reducible to each other.
\end{proposition}

\begin{proof}
Clearly, the feasibility problem is reducible to the optimization problem. We show now the converse implication and let us assume that we can decide in polynomial time whether a colorful linear program has a feasible solution. We explain then how to compute an optimal solution of~\eqref{P}.

We start by testing whether the problem~\eqref{P} has a feasible solution. If not, we are done. Otherwise, we compute a feasible solution by repeating several times the decision test: we apply it with the first color reduced to each point in turn, until the answer is `yes', keep this point, and proceed similarly for all colors one after the other. It provides an upper bound $\mu_{\sup}$ on the optimal value of~\eqref{P}. We compute also a lower bound $\mu_{\inf}$ by taking the noncolorful version of~\eqref{P}, i.e., by relaxing the constraint $$|\supp(\x)\cap I_i|\leq 1\qquad\forall i.$$ This lower bound is polynomially computable since it is usual linear programming.

Consider now the feasibility problem of deciding whether the following system has a solution. Using our polynomial decision test, we can solve it in polynomial time.
\begin{equation}\label{eq:sys_feas}
\left\{\begin{array}{ll}
\begin{pmatrix}  A&0\\ \c&1 \end{pmatrix}\x' = \begin{pmatrix} \b\\\mu\end{pmatrix}&\\ \\
\x'\geq \zero& \\ \\
|\supp(\x')\cap I_i|\leq 1 & \quad\forall i\\ \\
|\supp(\x')\cap \{n+1\}| \leq  1.&
\end{array}\right.
\end{equation}
The colors for this new colorful linear program are the $I_i$'s defined for the original program and an additional $(k+1)$th color $I_{k+1}=\{n+1\}$. The problem~\eqref{eq:sys_feas} is feasible if and only if a colorful solution $\begin{pmatrix}\x\\z\end{pmatrix}$ of this problem induces a colorful solution $\x$ of~\eqref{P} with value smaller than $\mu$.

 Using binary search, we can thus solve the problem. At each step, we consider the problem~\eqref{eq:sys_feas} with $\mu=(\mu_{\inf}+\mu_{\sup})/2$. If the problem is feasible, we update $\mu_{\sup}:=\mu$, otherwise we update $\mu_{\inf}:=\mu$. At each step the size of $[\mu_{\inf},\mu_{\sup}]$ is hence divided by two. Using Cramer's rule, we know that there is a minimal gap $\varepsilon$, polynomially computable, between two values of $\c\cdot\x$ with $\x$ being a feasible basis of the system $$A\x=\b,\; \x\geq\zero.$$ Hence, after at most $\log_2(\frac{\mu_{\sup}-\mu_{\inf}}{\varepsilon})$ iterations, the feasible basis we obtain for the problem with $\mu=\mu_{\sup}$ is the optimum.
\end{proof}

\subsection{Colorful diet}

Colorful linear programming can also be used as a modeling tool, as in the following example suggested to us by Max Klimm.

The diet problem has been introduced during the Second World War and aimed at defining the daily diet of U.S. soldiers. Given a set of nutriments and a set of foods, each containing a certain amount of each nutriment, the problem is to find an optimal diet with respect to some objective function, such that each nutriment is sufficiently provided. It was one of the first problems on which the simplex algorithm was tested (in 1947~\cite{Dan63}). Later, in 1990, Dantzig showed the limits of this model in an over-viewing paper~\cite{Dan90}, in which he described how he tried to apply the model to his own diet. The main struggle he encountered was ``the lack of variety'' of the solutions given by the model.
Adding bounds, he managed to avoid solutions using only one food, for instance bran\footnote{Dantzig finally followed his wife's advice, which was certainly much more efficient than linear programming.}. 
In this context, variety is thus achievable by standard linear programming technique. However, if he had known colorful linear programming, he would have had an additional modeling flexibility. It might be undesirable to have many foods from a same category in a solution, e.g., for logistic reasons. In general, one does not eat more than three types of vegetables, two types of meat and one fruit each day, and so on. Colorful linear programming consists in solving a linear program whose variables are partitioned into categories and with the additional constraints that the number of variables of each category used in a solution is bounded. It allows thus to avoid a too high variety inside the categories.

Formally, the original diet problem is the following. We are given $n$ foods and $m$ nutriments. Let $a_{ij}\in\mathbb R_+$ be the quantity of nutriment $j$ in one unit of food $i$, let $b_j$ be the quantity of nutriment $j$ needed daily, and let $r_i$ be the maximal amount of food $i$ that can be tolerated in one day. Finally, let $c_i$ be the cost of one unit of food $i$. We define the variables $x_1,\ldots,x_n\in \mathbb R_+$, modeling the quantity of food $i$ that will be recommended by the diet program. The diet problem aims at solving
$$ \begin{array}{rl}
\min &\c\cdot\x\\
\mbox{s.t.} &A\x\geq\b\\
& \x\leq\r \\
&\x\geq\zero,
\end{array}$$
where $A$ is the matrix $(a_{ij})$.

Now, assume that the foods are partitioned into different categories, i.e., $\{1,\ldots,n\}=I_1\cup\cdots\cup I_k$, and that it is required to have at most $\ell_h$ distinct foods selected in a category $I_h$. This requirement can be taken into account by adding the following constraints to the previous linear program:
$$|\supp(\x)\cap I_h|\leq \ell_h\qquad\forall h\in\{1,\ldots,k\}.$$ We get what we call the {\em colorful diet problem}.

By making $\ell_h$ copies of $I_h$ and by adding the corresponding $(x_{i,j})_{j\in[\ell_h]}$ variables and by using slack variables, this model can be readily written under the form \eqref{P}. Proposition~\ref{prop:feas_opt} shows that the colorful diet problem is polynomially reducible to {\sc Colorful Linear Programming} in its conic version.

\section{Complexity of colorful linear programming}\label{sec:complexity}

In this section, we come back to the original {\sc Colorful Linear Programming} problem, discuss quickly its NP-completeness, and give an application.

\subsection{Proof of \textup{NP}-completeness}

For a fixed $q\in\mathbb Z$, we define {\sc CLP}$(q)$ to {\sc Colorful Linear Programming} with the additional constraint that $k-d=q$.

\begin{lemma}\label{lem1}
If {\sc CLP}$(q)$ is \textup{NP}-complete, then {\sc CLP}$(q-1)$ is also \textup{NP}-complete.
\end{lemma}

\begin{proof} 
Let $\S_1,\ldots,\S_k$ in $\R^d$ be an instance with $k=d+q$. Define $d'=d+1$. Embedding this instance in $\R^{d'}$ by adding a $d'$th component equal to $0$, we get an instance with $k=d'+q-1$, every solution of which provides a solution for the case $k=d+q$, and conversely. This latter case being NP-complete, we get the conclusion.
\end{proof}

\begin{lemma}\label{lem2}
If {\sc CLP}$(q)$ is \textup{NP}-complete, then {\sc CLP}$(q+1)$ is also \textup{NP}-complete.
\end{lemma}

\begin{proof}
Let $\S_1,\ldots,\S_k$ in $\R^d$ be an instance with $k=d+q$. Define $d'=d+1$ and $k'=k+2$. Embed this instance in $\R^{d'}$ by adding a $d'$th component equal to $0$. Add two sets $\S_{k+1}$ and $\S_{k+2}$ entirely located at coordinate $(0,\ldots,0,1)$. We have thus an instance with $k'=d'+q+1$, every solution of which provides a solution for the case $k=d+q$, and conversely. This latter case being NP-complete, we get the conclusion.
\end{proof}

We have the following proposition.

\begin{proposition}\label{prop:clpq}
{\sc CLP}$(q)$ is {\sc NP}-complete for any fixed $q\in\mathbb{Z}$.
\end{proposition}

In particular, it is NP-complete for $q=1$.

\begin{proof}[Proof of Proposition~\ref{prop:clpq}]
{\sc CLP}$(0)$ is NP-complete according to Theorem 6.1 in~\cite{BO97}. Lemmas~\ref{lem1} and~\ref{lem2} allow to conclude.
\end{proof}

Polynomially checkable sufficient conditions ensuring the existence of a positively dependent colorful set exist: the condition of the colorful \cara{} theorem is one of them. More general polynomially checkable sufficient conditions when $k=d+1$ are given in~\cite{AB+09,HPT08,MD12}. However, the fact that CLP$(1)$ is NP-complete implies that there are no polynomially checkable conditions that are simultaneously sufficient and necessary for a positively dependent colorful set to exist when $k=d+1$, unless P=NP.

\begin{remark}\label{rk:gen}
The instances built in the proof of Lemma~\ref{lem2} are not in general position, since $\zero$ and the $\S_i$'s with $i\leq k$ are all in a same hyperplane. 
We could wonder whether the case $k=d+1$ remains NP-complete when the points are in general position. The answer is yes, and we explain how to reduce the instance built in the proof of Lemma~\ref{lem2} to an instance in general position.

First, the sets $\S_{k+1}$ and $\S_{k+2}$ can be slightly perturbed without changing the conclusion. Second, we slightly move $\zero$ into one of the halfspaces delimited by the hyperplane containing the $\S_i$'s for $i\leq k$. We choose the halfspace containing $\S_{k+1}$ and $\S_{k+2}$. This move must be sufficiently small so that $\zero$ does not traverse another hyperplane generated by $d'$ points in $\bigcup_{i=1}^{k+2}\S_i$. All coordinates being rational, Cramer's formula allows to compute a length of the displacement that ensures this condition. Third, we move each point of the $\bigcup_{i=1}^{k}\S_i$ independently along a line originating from $\zero$.
\end{remark}


\subsection{Projection and colorful linear programming}

Algorithmic questions related to projecting polytopes are usually identified as difficult questions. Tiwary~\cite{Tiw12} recently showed that given two polytopes $\mathcal{Q}$ and $\mathcal{Q}'$ described by systems of linear inequalities, deciding whether $\mathcal{Q}$ is a projection of $\mathcal{Q}'$ is coNP-complete. Note that it is in coNP since deciding whether a partial solution of a system of linear inequalities can be extended to a full solution is a linear programming problem. His proof of coNP-completeness uses a reduction of the problem of deciding whether a polytope described by its facets is contained in a polytope described by its vertices, which is a coNP-complete problem~\cite{FrOr85}. {\sc Colorful Linear Programming} is another way to prove this result. 

Take any instance $\S_1,\ldots,\S_d,\boldsymbol{p}$ of the conic version of {\sc Colorful Linear Programming}, all points being in general position, with $\conv(\{\boldsymbol{p}\}\cup\bigcup_{i=1}^{d}\S_i)$ not containing $\zero$. Because of Proposition~\ref{prop:clpq} and Remark~\ref{rk:gen}, the problem of deciding whether there is colorful solution is NP-complete. We show how to reduce this instance to an instance of the aforementioned polytope projection problem. We define $A_i$ to be the matrix with the columns being the vectors in $\S_i$, for $i=1,\ldots,d$. We define then the following polytopes: $$\mathcal{P}=\left\{\boldsymbol{x}=(\boldsymbol{x}_1,\ldots,\boldsymbol{x}_d)\in\R_+^{d^2}:\,\sum_{i=1}^dA_i\boldsymbol{x}_i=\boldsymbol{p}\right\}$$ and $$\mathcal{P}_i=\left\{\boldsymbol{x}=(\boldsymbol{x}_1,\ldots,\boldsymbol{x}_d)\in\mathcal{P}:\,\boldsymbol{x}_i=\zero\right\}.$$ They are polytopes because of the assumption $\conv(\{\boldsymbol{p}\}\cup\bigcup_{i=1}^{d}\S_i)$ does not contain $\zero$.

There exists a colorful solution to {\sc Colorful Linear Programming} problem we consider if and only if $\mathcal{P}\setminus\conv(\bigcup_{i=1}^d\mathcal{P}_i)$ is nonempty. Indeed, if there exists a colorful solution to the {\sc Colorful Linear Programming} problem, the latter set in nonempty: a colorful solution provides a point (a basis in the linear programming terminology) in $\mathcal{P}$ with each $\boldsymbol{x}_i$ being nonzero, because of the general position assumption. Conversely, if the set is nonempty, there is a vertex of $\mathcal{P}$ not in $\conv(\bigcup_{i=1}^d\mathcal{P}_i)$, and such a vertex has exactly $d$ nonzero components, each corresponding to a column of a distinct $A_i$, and provides a solution to the {\sc Colorful Linear Programming} problem we consider.

Deciding whether $\mathcal{P}\setminus\conv(\bigcup_{i=1}^d\mathcal{P}_i)$ is nonempty is therefore NP-hard. We prove below that $\conv(\bigcup_{i=1}^d\mathcal{P}_i)$ is a projection of some higher dimensional polytope $\mathcal Q'$. Hence, deciding whether $\mathcal{P}\setminus\conv(\bigcup_{i=1}^d\mathcal{P}_i)$ is nonempty is equivalent to deciding whether $\mathcal{P}$ is the projection of $\mathcal Q'$. 

The polytope $\conv(\bigcup_{i=1}^d\mathcal{P}_i)$ is described by the solutions $\boldsymbol{x}=(\boldsymbol{x}_1,\ldots,\boldsymbol{x}_d)$ satisfying the following system of linear equalities and inequalities:
$$\left\{\begin{array}{ll}\ds{\sum_{j=1}^dA_i\boldsymbol{x}_{ij}-y_i\boldsymbol{p}} = \zero & \quad\forall i \\
\ds{\sum_{i=1}^d\boldsymbol{x}_{ij}} = \x_j & \quad\forall j \\
\ds{\sum_{i=1}^dy_i} = 1 
\\ \x_{ii} = \zero & \quad\forall i\\
\y\geq\zero & \quad\forall i\\
\x_{ij}\geq\zero & \quad\forall i,j.
\end{array}\right.$$

Indeed, a point $\boldsymbol{x}=(\boldsymbol{x}_1,\ldots,\boldsymbol{x}_d)\in\conv(\bigcup_{i=1}^d\mathcal P_i)$ is such that $\boldsymbol{x}=\sum_{i=1}^dy_i\boldsymbol{x}'_{i}$, with $\sum_{i=1}^dy_i=1$ and $\boldsymbol{x}'_i=(\boldsymbol{x}'_{i1},\ldots,\boldsymbol{x}'_{id})\in\mathcal P_i$ for each $i$. Defining $\boldsymbol x_{ij}$ to be $y_i\boldsymbol{x}'_{ij}$ shows that such an $\boldsymbol x$ satisfies the system. Conversely, a solution of the system induces a point $\boldsymbol{x}$ that can be written as $\sum_{i=1}^dy_i\boldsymbol{x}'_i$ with $\boldsymbol{x}'_i\in\mathcal P_i$ for all $i$. Indeed, define $\boldsymbol{x}'_{ij}=\frac{1}{y_i}\boldsymbol{x}_{ij}$ when $y_i\neq 0$, and $\boldsymbol{x}'_{ij}=\zero$ otherwise. In this latter case, all the $\boldsymbol{x}_{ij}$'s are equal to $\zero$ because of the assumption $\zero\notin\conv(\{\boldsymbol{p}\}\cup\bigcup_{i=1}^{d+1}\boldsymbol{S}_i)$.

\section{Special cases and analogues of colorful linear programming in combinatorics}\label{sec:combin}

We give in this section two combinatorial corollaries of the colorful \cara{} theorem. For each of them, we provide a direct proof. In both cases, we show that the colorful set can be computed in polynomial time, and get thus special polynomial cases of \ccp.

\begin{proposition}
Let $D=(V,A)$ be a directed graph with $n$ vertices. Let $C_1,\ldots,C_n$ be pairwise arc-disjoint circuits. Then there exists a circuit $C$ sharing at most one arc with each of these $C_i$. Moreover, such a circuit can be computed in polynomial time.
\end{proposition}

The existence of the colorful circuit as a consequence of the colorful \cara{} theorem has already been noted and is attributed to Frank and Lov\'asz~\cite{Bar82}.

\begin{proof}
We consider the bipartite graph with vertex classes $V$ and $\{1,\ldots,n\}$, and in which edge $vi$ exists if the vertex $v$ belongs to $C_i$. If each $X\subseteq V$ touches at least $|X|$ distinct colors, Hall's marriage theorem ensures the existence of a perfect matching in the bipartite graph. We can thus select for each vertex $v\in V$ an arc $a$ in $\delta^{-}(v)$ belonging to a distinct $C_i$. The subgraph induced by these arcs contains a circuit $C$ as required.

Otherwise, there is a subset $X\subseteq V$ with a neighborhood in the bipartite graph of cardinality at most $|X|-1$. One can remove $X$ from $D$ and apply induction. Note that the existence of such an $X$ can be decided in polynomial time by a classical maximum matching algorithm, which provides also the set $X$ itself if it exists.
\end{proof}

The existence statement of the next proposition is a consequence of the conic version of the colorful \cara{} theorem. We provide a direct proof based on a greedy algorithm.

\begin{proposition}\label{prop:path}
Let $D=(V,A)$ be a directed graph with $n$ vertices. Let $s$ and $t$ be two vertices, and $P_1,\ldots,P_{n-1}$ be pairwise arc-disjoint $s$-$t$ paths. Then there exists an $s$-$t$ path $P$ sharing at most one arc with each $P_i$. Moreover, such a path can be computed in polynomial time.
\end{proposition}

\begin{proof}
We build progressively an arborescence rooted at $s$. We start with $X=\{s\}$. At each step, $X$ is the set of vertices reachable from $s$ in the partial arborescence. At step $i$, if $X$ does not contain $t$, choose an arc $a$ of $P_i$ belonging to $\delta^+(X)$ and add to $X$ the endpoint of $a$ not yet in $X$. This arc exists since by direct induction $X$ is of cardinality $i$ at step $i$ and the $s$-$t$ path $P_i$ leaves $X$. 
\end{proof}

We end the paper by a brief survey of matroidal counterparts of colorful linear programming. The next proposition is common knowledge in combinatorics. It is a matroidal version of the colorful \cara{} theorem (with an additional algorithmic result).

\begin{proposition}
Let $M$ be a matroid of rank $d$. Assume that the elements of $M$ are colored in $d$ colors. If there exists a monochromatic basis in each color, then there exists a colorful basis and this latter can be found by a greedy algorithm.
\end{proposition}

A matroidal version of the Octahedron lemma stated in Section~\ref{subsec:octahedral} also exists. It is due to Magnanti~\cite{Ma74}.

\begin{proposition}\label{prop:matroid_another}
Let $M$ be a matroid of rank $d$ with no loops. Assume that the elements of $M$ are colored in $d$ colors and that the number of elements colored in each color is at least two. If there is a colorful basis, then there is another colorful basis and this latter can then be found in polynomial time.
\end{proposition}

The proof by Magnanti is based on the matroid intersection algorithm due to Lawler~\cite{La75}. The same algorithm shows that the matroidal version of {\sc Colorful Linear Programming}, namely deciding whether there is a colorful basis in a matroid whose elements are colored, is polynomial.

\begin{remark}[Back to Sperner's lemma]\label{rk:sperner_comp} Remark~\ref{rk:sperner} of Section~\ref{subsec:bimatrix} shows that even a very special case of Sperner's lemma already leads to a PPAD-complete problem. The matroidal counterpart of the Octahedron lemma implies that the problem becomes polynomial when the triangulation is the boundary of the {\em cross-polytope}. The cross-polytope is the convex hull of the vectors of the standard orthonormal basis and their negatives.
\end{remark}

\begin{proposition}
Let $\mathsf{T}$ be the boundary of the $(d+1)$-dimensional cross-polytope and let $\lambda:V(\mathsf{T})\rightarrow\{1,\ldots,d+1\}$ be any labeling. Assume given a fully-labeled simplex. Another fully-labeled simplex can be computed in polynomial time.
\end{proposition}

\begin{proof} If a vertex has a label that appears only once on $V(\mathsf{T})$, we remove it and its antipodal, and work on the boundary of a cross-polytope with a dimension smaller by one. Solving this new problem leads to a solution for the original problem. We repeat this process until each label appears exactly twice.  Now, note that the simplices of the boundary of a cross-polytope form the independent sets of a matroid (it is a partition matroid). Considering the labels as colors, the conclusion follows then from Proposition~\ref{prop:matroid_another}.
\end{proof}

With a similar proof (omitted), we also have the following proposition.

\begin{proposition}
Let $\mathsf{T}$ be the boundary of the $(d+1)$-dimensional cross-polytope and let $\lambda:V(\mathsf{T})\rightarrow\{1,\ldots,d+1\}$ be any labeling. Deciding whether there is a fully-labeled simplex can be done in polynomial time. Moreover, if there is such a fully-labeled simplex, it can be found in polynomial time as well.
\end{proposition}

\bigskip

\noindent{\bf Acknowledgments.} The authors thank the reviewers for their helpful comments.

\bibliographystyle{amsplain}
\bibliography{refs}

\end{document}